\pdfoutput=1
\documentclass[conference]{IEEEtran}
\usepackage[utf8]{inputenc}
\usepackage[T1]{fontenc}
\usepackage[noadjust]{cite}
\usepackage{datetime}
\usepackage{paralist}
\usepackage{cite}
\usepackage{tikz}
\usetikzlibrary{positioning,calc,through,fit,%
  decorations.pathreplacing,matrix,decorations.markings,%
  backgrounds,arrows,shadows}
\usepackage{amsmath,amsthm}
\usepackage{amssymb,amsthm}
\usepackage{mathtools}

\usepackage{algorithm}
\usepackage[noend]{algpseudocode}

\usepackage{stmaryrd}
\usepackage[multiuser,layout=margin]{fixme}
\usepackage{xspace}
\usepackage{hyperref}
\usepackage{url}
\usepackage{microtype}
\input{define.tex}
\draftfalse

\FXRegisterAuthor{mc}{amc}{MC}
\FXRegisterAuthor{cp}{acp}{CP}
\revision$LastChangedRevision: 1186 $

\title{A Crevice on the Crane Beach:\\
  Finite-Degree Predicates}

\author{\IEEEauthorblockN{Michaël Cadilhac}
  \IEEEauthorblockA{WSI, Universität Tübingen\\
    Email: \texttt{michael@cadilhac.name}}
  \and
  \IEEEauthorblockN{Charles Paperman}
  \IEEEauthorblockA{WSI, Universität Tübingen\\
    Email: \texttt{charles.paperman@gmail.com}}}

\begin{document}

\maketitle

\begin{abstract}
  First-order logic (FO) over words is shown to be equiexpressive with FO
  equipped with a restricted set of numerical predicates, namely the order, a
  binary predicate $\MSBz*$, and the \emph{finite-degree} predicates:
  $\FO*[\arb] = \FO*[\leq, \MSBz, \fin]$.
  The Crane Beach Property (CBP), introduced more than a decade ago, is true of
  a logic if all the expressible languages admitting a neutral letter are
  regular.
  Although it is known that $\FO*[\arb]$ does not have the CBP, it is shown here
  that the (strong form of the) CBP holds for both $\FO*[\leq, \fin]$ and
  $\FO*[\leq, \MSBz]$.  Thus $\FO*[\leq, \fin]$ exhibits a form of locality and
  the CBP, and can still express a wide variety of languages, while being one
  simple predicate away from the expressive power of $\FO*[\arb]$.  The counting
  ability of $\FO*[\leq, \fin]$ is studied as an application.
\end{abstract}

\IEEEpeerreviewmaketitle

\section{Introduction}

Ajtai~\cite{ajt83} and Furst, Saxe, and Sipser~\cite{furst-saxe-sipser84} showed
some 30 years ago that \parity, the language of words over $\{0, 1\}$ having an
even number of $1$, is not computable by families of shallow circuits, namely
\ACz circuits.  Since then, a wealth of precise expressiveness properties of
\ACz has been derived from this sole
result~\cite{barrington-compton-straubing-therien92,straubing94}.  Naturally
aiming at a better understanding of the core reasons behind this lower bound, a
continuous effort has been made to provide alternative proofs of
$\parity \notin \ACz$.  However, this has been a rather fruitless endeavor, with
the notable exception of the early works of Razborov~\cite{raz87} and
Smolenski~\cite{smo87} that develop a less combinatorial approach with an
algebraic flavor.  For instance, Koucký~et~al.~\cite{kopolath06} foray into
descriptive complexity and use model-theoretic tools to obtain
$\parity \notin \ACz$, but assert that ``contrary to [their] original hope,
[their] Ehrenfeucht-Fraïssé game arguments are not simpler than classical lower
bounds.''  More recent promising approaches, especially the topological ones
of~\cite{Gehrke2010,Czarnetzki2016}, have yet to yield strong lower bounds.

A different take originated from a conjecture of Lautemann and Thérien,
investigated by Barrington~et~al.~\cite{baimlascth05}: the Crane Beach
Conjecture.  They noticed that the letter $0$ acts as a \emph{neutral letter}
in \parity, i.e., $0$ can be added or removed from any word without affecting
its membership to the language.  If a circuit family recognizes a language with
a neutral letter, it seems convincing that the circuits for two given input
sizes should look very similar, that is: the circuit family must be highly
\emph{uniform}.  It was thus conjectured that all neutral letter languages in
\ACz were regular, and this was disproved in~\cite{baimlascth05}.

This however sparked an interest in the study of neutral letter languages, in
particular from the descriptive complexity view.  Indeed, \ACz circuits
recognize precisely the languages expressible in $\FO[\arb]$, where $\arb$
denotes all possible \emph{numerical predicates} (expressing numerical
properties of the positions in a word).  Further, as all regular neutral letter
languages of $\FO[\arb]$ are star-free~\cite{baimlascth05}, i.e.,
in~$\FO[\leq]$, the Crane Beach Conjecture asked:

\centerline{\emph{Are all neutral letter languages of $\FO[\arb]$ in $\FO[\leq]$?}}

Note that this echoes the above intuition on uniformity, since the numerical
predicates correspond precisely to the allowed power to compute the circuit for
a given input length~\cite{bela06}.  The intuition on the logic side is even more
compelling: if a letter can be introduced anywhere without impacting membership,
then the only meaningful relation that can relate positions is the linear order.
However, first-order logic can ``count'' up to $\log n$ (see, e.g.,
\cite{durand07}), meaning that even within a word with neutral letters,
$\FO[\arb]$ can assert some property on the number of nonneutral letters.  This
is, in essence, why nonregular neutral letter languages can be expressed in
$\FO[\arb]$.

In the recent years, a great deal of efforts was put into studying the Crane
Beach Property in different logics, i.e., whether the definable neutral letter
languages are regular.  Krebs and Sreejith~\cite{krebs12}, building on the work
of Roy and Straubing~\cite{roy-straubing07}, show that all first-order logics
with monoidal quantifiers and $+$ as the sole numerical predicate have the Crane
Beach Property.  Lautemann~et~al.~\cite{laut06} show Crane Beach Properties for
classes of bounded-width branching programs, with an algebraic approach relying
on communication complexity.  Some expressiveness results were also derived from
Crane Beach Properties, for instance Lee~\cite{lee03} shows that $\FO[+]$ is
strictly included in $\FO[\leq, \times]$ by proving that only the former has the
Crane Beach Property.  Notably, all these logics are quite far from full
$\FO[\arb]$, and in that sense, fail to identify the part of the arbitrary
numerical predicates that fit the intuition that they are rendered useless by
the presence of a neutral letter.

In the present paper, we identify a large class of predicates, the
\emph{finite-degree} predicates, and a predicate $\MSBz$ such that \emph{any
  numerical predicate} can be first-order defined using them and the order; in
symbols, $\FO[\leq, \MSBz, \fin] = \FO[\arb]$.  We show that, strikingly, both
$\FO[\leq, \MSBz]$ and $\FO[\leq, \fin]$ have the Crane Beach Property, this
latter statement being our main result.  Hence showing that some nonregular
neutral letter language is not expressible in $\FO[\arb]$ could be done by
showing that $\MSBz$ may be removed from any $\FO[\leq, \MSBz, \fin]$ formula
expressing it.

The proof for the Crane Beach Property of $\FO[\leq, \fin]$ relies on a
communication complexity argument different from that of~\cite{laut06}.  It is
also unrelated to the database collapse techniques of~\cite{baimlascth05}
(succinctly put, no logic with the Crane Beach Property has the so-called
\emph{independence property}, i.e., can encode arbitrary large sets).  We will
show that in fact $\FO[\leq, \fin]$ \emph{does} have the independence property.
This provides, to the best of our knowledge, the first example of a logic that
exhibits \emph{both} the independence and the Crane Beach properties.

The aforementioned counting property of $\FO[\arb]$ led to the
conjecture~\cite{baimlascth05,lee03} that a logic has the Crane Beach Property
\emph{if and only if} it cannot count beyond a constant.  To the best of our
knowledge, neither of the directions is known; we show however that
$\FO[\leq, \fin]$ can only count up to a constant, by showing that it cannot
even express very restricted forms of the addition. This adds evidence to the
``if'' direction of the conjecture.

\emph{Structure of the paper.}\quad In Section~\ref{sec:prelim}, we introduce
the required notions, although some familiarity with language theory and logic
on words is assumed (see, e.g., \cite{straubing94}).  In
Section~\ref{sec:foarb}, we show that $\FO[\leq, \MSBz, \fin] = \FO[\arb]$.  In
Section~\ref{sec:fomsb}, we present a simple proof, relying on a much harder
result from~\cite{baimlascth05}, that $\FO[\leq, \MSBz]$ has the Crane Beach
Property.  The failing of the aforementioned collapse technique for
$\FO[\leq, \fin]$ is shown in Section~\ref{sec:finindep}.  We tackle the Crane
Beach Property of $\FO[\leq, \fin]$, our main result, in
Section~\ref{sec:cbcfin}, after the necessary tools have been developed.
Finally, in Section~\ref{sec:count}, we focus on the counting inabilities of
$\FO[\leq, \fin]$.

\emph{Previous works.}\quad Finite-degree predicates were introduced by the
second author in~\cite{paperman15}, in the context of two-variable logics.
Therein, it is shown that the two-variable fragment of $\FO[\leq, \fin]$ has the
Crane Beach Property, and, even stronger, that the neutral letter languages
expressible with $k$ quantifier alternations  can be expressed without the finite-degree
predicates with the \emph{same} amount of quantifier alternations.  The techniques used
in~\cite{paperman15} are specific to two-variable logics, relying heavily on the
fact that each quantification depends on a \emph{single} previously quantified
variable.  We thus stress that the communication complexity argument developed
in Section~\ref{sec:cbcfin} is unrelated to~\cite{paperman15}.
 
The fact that two sets of predicates can both verify the Crane Beach Property
while their union does not has already been witnessed in~\cite{baimlascth05}.
Indeed, letting $\mon$ be the set of \emph{monoidal} numerical predicates, the
Property holds for both $\FO[\leq, +]$ and $\FO[\leq, \mon]$ but fails for
$\FO[\leq, +, \mon]$, although this latter class is less expressive than
$\FO[\arb]$ (this can be shown using the same proof
as~\cite[Proposition~5]{kopolath06}).

\section{Preliminaries}\label{sec:prelim}
\subsection{Generalities} We write $\bbn = \{0, 1, 2, \ldots\}$ for the set of
nonnegative numbers.  For $n \in \bbn$, we let $[n] = \{0, 1, \ldots, n - 1\}$.
A function $f\colon \bbn \to \bbn$ is \emph{nondecreasing} if $m > n$ implies
$f(m) \geq f(n)$.

An alphabet $A$ is a finite set of letters (symbols), and we write $A^*$ for
the set of finite words.  For $u = u_0u_1\cdots u_{n-1}$, the length $n$ of $u$
is denoted $|u|$.  We write $\eps$ for the empty word and $A^{\leq k}$ for words
of length $\leq k$.

\subsection{Logic on words} For an alphabet $A$, let $\sigma_A$ be the
vocabulary $\{\lt a \mid a \in A\}$ of unary \emph{letter predicates}.  A
(finite) word $u = u_0u_1\cdots u_{n-1} \in A^*$ is naturally associated with
the structure over $\sigma_A$ with universe $[n]$ and with $\lt a$ interpreted
as the set of positions $i$ such that $u_i = a$, for any $a \in A$.  A
\emph{numerical predicate} is a $k$-ary relation symbol together with an
interpretation in $[n]^k$ for each possible universe size $n$.  Given a formula
$\phi$ that relies on some numerical predicates and a word $u$, we write
$u \models \phi$ to mean that $\phi$ is true of the $\sigma_A$-structure for $u$
augmented with the interpretations of the numerical predicates for the universe
of size $|u|$.  A formula $\phi$ thus \emph{defines} or \emph{expresses} the
language $\{u \in A^* \mid u \models \phi\}$.

\subsection{Classes of formulas} We let $\arb$ be the set of all numerical
predicates.  Given a set $\calN \subseteq \arb$, we write $\FO[\calN]$ for the
set of first-order formulas built using the symbols from $\calN \cup \sigma_A$,
for any alphabet $A$.  Similarly, $\MSO[\calN]$ denotes monadic second-order
formulas built with those symbols.  We further define the quantifiers $\Maj$ and
$\exists^\equiv_i$, for $i \in \bbn$, that will only be used in discussions:
\begin{itemize}
\item $u \models (\Maj x)[\phi(x)]$ iff there is strict majority of positions
  $i \in [|u|]$ such that $\langle u, x := i\rangle \models \phi$;
\item $u \models (\exists^\equiv_i)[\phi(x)]$ iff the number of positions $i \in
  [|u|]$ verifying $\langle u, x := i\rangle \models \phi$ is a multiple of $i$.
\end{itemize}
We will write $\MAJ[\calN]$ and $\FOMAJ[\calN]$ with the obvious meanings.
Further, $\FOMOD[\calN]$ allows all the quantifiers $\exists^\equiv_i$ in
$\FO[\calN]$ formulas.

\subsection{On numerical predicates} 

The most ubiquitous numerical predicate here will be the binary order predicate
$\leq$.  The predicate that zeroes the most significant bit (MSB) of a number
will also be important: $(m, n) \in \MSBz$ iff
$n = m - 2^{\lfloor \log m\rfloor}$.  Note that both predicates do \emph{not}
depend on the universe size, and we single out this concept:
\begin{definition}
  A $k$-ary numerical predicate $P$ is \emph{unvaried} if there is a set
  $E \subseteq \bbn^k$ such that the interpretation of $P$ on universes of size
  $n$ is $E \cap [n]^k$.  In this case, we identify $P$ with the set $E$.  It is
  \emph{varied} otherwise.\footnote{The relevance of this concept has been noted
    in previous works (e.g., ~\cite{baimlascth05}), but was left unnamed.  The
    second author used in~\cite{paperman15} the terms \emph{(non)uniform}, an
    unfortunate coinage in this context.  We prefer here the less conflicting
    terms \emph{(un)varied}.}  We write $\arbu$ for the set of unvaried
  numerical predicates.
\end{definition}
Naturally, any varied predicate can be converted to an unvaried one by turning
the universe length into an argument and quantifying the maximum position; this
implies in particular that $\FO[\arb] = \FO[\arbu]$.  This is however not
entirely innocuous, as will be discussed in Section~\ref{sec:count}.

We will rely on the following class of unvaried predicates, generalizing a
definition of~\cite{paperman15} (see also the older notion of ``finite
formula''~\cite{kps86}):
\begin{definition}
  An unvaried predicate $P \subseteq \bbn^k$ is of \emph{finite
    degree}\footnote{The name stems from the fact that the hypergraph defined by
    $P$, with edges of size $k$, is of finite degree.} if for all $n \in \bbn$,
  $n$ appears in a finite number of tuples in $P$.  We write $\fin$ for the
  class of such predicates.
\end{definition}
Note that this does \emph{not} imply that there is a $N$ that bounds the number
of appearance for all $n$'s.  Some examples:
\begin{itemize}
\item $\MSBz$ is \emph{not} a finite-degree predicate, as, e.g., $(2^n, 0) \in
  \MSBz$ for any $n$, hence $0$ appears infinitely often;
\item Any \emph{unvaried} monadic numerical predicate is of finite degree, this
  implies in particular that any language over a unary alphabet is expressed by
  a $\FO[\leq, \fin]$ formula;
\item The graph of any nondecreasing unbounded function $f\colon \bbn \to \bbn$
  defines a finite-degree predicate, since $f^{-1}(n)$ is a finite set for all
  $n$;
\item The order, sum, and multiplication are not of finite degree;
\item One can usually ``translate'' unvaried predicates to make them finite
  degree; for instance, the predicate true of $(x,y)$ if $y-x < x < y$ is of
  finite degree, see also the proof of Proposition~\ref{prop:count}.
\end{itemize}



\subsection{Crane Beach Property} A language $L \subseteq A^*$ is said to have a
neutral letter if there is a $e \in A$ such that adding or deleting $e$ from a
word does not change its membership to $L$.
Following~\cite{laut06}, we say that a logic has the Crane
Beach Property if all the neutral letter languages it defines are regular.  We
further say that it has the \emph{strong} Crane Beach Property if all the
neutral letter languages it defines can be defined using order as the sole
numerical predicate.

\section{\(\FO*[\arb]\) and \(\FO*[\leq, \MSBz, \fin]\) define the same
  languages}\label{sec:foarb}

In this section, we express all the numerical predicates using only
finite-degree ones, $\MSBz$, and the order.  The result is a variant of
\cite[Theorem~3]{paperman15}, where it is proven for the two-variable fragment,
and on neutral letter languages.
\begin{theorem}\label{thm:fo}
  $\FO[\arb]$ and $\FO[\leq, \MSBz, \fin]$ define the same languages.
\end{theorem}
\begin{proof}
  We show that any $\FO[\arbu]$ language is definable in
  $\FO[\leq, \MSBz, \fin]$.
  
  The main idea is to divide the set of word positions in four contiguous zones
  and have the variables range over only the second zone, called the \emph{work
    zone}.  Given an input of length $\ell=2^n$, the set of positions $[\ell]$
  is divided in four zones of equal size $2^{n-2}$; if the input length is not a
  power of $2$, then we apply the same split as the closest greater power of
  two, leaving the third and fourth zone possibly smaller than the first two.

  As an example, suppose that the word size is $\ell = 11110$ (here and in the
  following, we write numbers in binary).  The four zones of $[\ell]$ will be:
  \def\elt#1{\text{#1)~~}}
  \begin{align*}
    \elt{1} 00000 \to 00111; &\qquad \elt{2} 01000 \to 01111;\\
    \elt{3} 10000 \to 10111; &\qquad \elt{4} 11000 \to 11101=\ell - 1\enspace.
  \end{align*}
  
  The work zone has two salient properties: 1. Checking that a number
  $k \in [\ell]$ belongs to it amounts to checking that $k$ has exactly one
  greater power of two; in particular, two work-zone positions share the same
  MSB; 2. Any number in $[\ell]$ outside the work zone can be obtained by
  replacing the MSB of a number in the work zone with some other bits (0, 10,
  and 11, for the first, third, and fourth zone, respectively); we call this a
  \emph{translation to a zone}, e.g., in our example above, $10101$ is the
  translation of $01101$ to the third zone.

  \def\work{\mathsf{work}}%
  \def\trans{\mathsf{trans}}%
  More formally, we can define a formula $\work(x)$ which is true iff $x$
  belongs to the work zone, by expressing that there is exactly one power of two
  strictly greater than $x$, using the monadic predicate true on powers of two.
  Moreover, we can define formulas $\trans^{(i)}(x, y)$, $1 \leq i \leq 4$,
  which are true if $x$ is in the work zone and $y$ is its translation to the
  $i$-th zone; let us treat the case $i = 3$, the others being similar.  The
  formula $\trans^{(3)}(x, y)$ is true if $y$ is obtained by replacing the MSB
  of $x$ with $10$, this is expressed using $\MSBz$ by finding $z$ such that
  $\MSBz(x, z)$ holds and then checking that $y$ is the first value $z'$
  strictly greater than $x$ such that $\MSBz(z', z)$ holds.

  The strategy will then be to: 1. Quantify over the work zone only; 2. Modify
  the predicates to internally change the MSBs according to which zone the
  variables were supposed to belong; 3. Compute the translations of the
  variables for the letter predicates.  Step 1 relies on $\work$ and
  $\trans^{(i)}$, Step 2 transforms all numerical predicates to finite-degree
  ones, and Step 3 simply uses $\trans^{(i)}$.

  Let $\phi \in \FO[\arbu]$. \emph{Step 1.}  We rewrite $\phi$ with
  \emph{annotated} variables; with $x$ a variable, we write $x^{(i)}$,
  $1 \leq i \leq 4$, to mean ``$x$ translated to zone $i$''---as all the
  variables will be quantified in the work zone, this is well defined.
  The following rewriting is then performed:
  \begin{align*}
    &\exists x\;\psi(x) \leadsto \\
    &\qquad \exists x \Big[\work(x) \land \bigvee_{1 \leq i
      \leq 4}\Big[
      (\exists y)[\trans^{(i)}(x, y)] \land \psi(x^{(i)}) \big] \Big]\enspace,
  \end{align*}
  and \emph{mutatis mutandis} for $\forall$.

  \emph{Step 2.} We sketch this step for binary numerical predicates.  Suppose
  such a predicate $P$ is used in $\phi$.  For $1 \leq i, j \leq 4$, we define
  the predicate $P^{(i, j)}$ that expects two work-zone positions, translates
  them to the $i$-th and $j$-th zone, respectively, then checks whether they
  belong to $P$.  Crucially, as the inputs are work-zone positions, $P^{(i,j)}$
  immediately rejects if they do not share the same MSB: it is thus a
  finite-degree predicate.  Now every occurrence of $P(x^{(i)}, y^{(j)})$ in
  $\phi$ can be replaced by $P^{(i,j)}(x, y)$.

  \emph{Step 3.} The only remaining annotated variables appear under letter
  predicates.  To evaluate them, we simply have to retrieve the translated
  position.  Hence each $\lt{a}(x^{(i)})$ will be replaced by
  $(\exists y)[\trans^{(i)}(x, y) \land \lt a(y)]$, concluding the proof.
\end{proof}

\begin{remark}
  Theorem~\ref{thm:fo} can be shown to hold also for
  $\FOMAJ[\leq, \MSBz, \fin]$, i.e., this logic is equiexpressive with
  $\FOMAJ[\arb]$.  The main modification to the proof is to allow
  \emph{arbitrary} quantifications (as opposed to work zone ones only) and
  \emph{compute} the work zone equivalent of each position before checking the
  numerical predicates.  This ensures that the number of positions verifying a
  formula is not changed.  Likewise, $\FOMOD[\leq, \MSBz, \fin]$ is equivalent
  with $\FOMOD[\arb]$.
\end{remark}

\section{\(\FO[\leq, \MSBz]\) has the Crane Beach Property}\label{sec:fomsb}

Following a short chain of rewriting, we will express $\MSBz$ using predicates
that appear in~\cite{baimlascth05} and conclude that:
\begin{theorem}
  $\FO[\leq, \MSBz]$ has the strong Crane Beach Property.
\end{theorem}
\begin{proof}
  Let $f\colon \bbn \to \bbn$ be defined by
  $f(n) = 2^{\left(\lfloor \log n\rfloor^2\right)}$, and let $F\subseteq \bbn^2$ be its graph.
  Barrington~et~al.~\cite[Corollary 4.14]{baimlascth05} show that $\FO[\leq, +, F]$ has
  the strong Crane Beach Property; we show that $\MSBz$ can be expressed in that
  logic.  First, the monadic predicate $Q = \{2^n \mid n \in \bbn\}$ is
  definable in $\FO[\leq, F]$, since $n$ is a power of two iff
  $f(n-1) \neq f(n)$.  Second, given $n \in \bbn$, the greatest power of two
  smaller than $n$ is $p = 2^{\lfloor \log n\rfloor}$, which is easy to find in
  $\FO[\leq, Q]$.  Finally, $\MSBz(n, m)$ is true iff $m + p = n$, and is thus
  definable in $\FO[\leq, +, F]$.
\end{proof}

\begin{remark}
  From Lange~\cite{lange04}, $\MAJ[\leq]$ and $\FOMAJ[\leq, +]$ are
  equiexpressive, and as $\MSBz$ is expressible using the unary predicate
  $\{2^n \mid n \in \bbn\}$ and the sum, this shows that $\MAJ[\leq, \fin]$ is
  equiexpressive with $\FOMAJ[\arb]$.  Hence $\MAJ[\leq, \fin]$ does \emph{not}
  have the strong Crane Beach Property.
\end{remark}

\section{\(\FO*[\leq, \fin]\) has the Independence
  Property}\label{sec:finindep}

In~\cite{baimlascth05}, an important tool is introduced to show Crane Beach
Properties, relying on the notion of \emph{collapse} in databases,
see~\cite[Chapter~13]{libkin04} for a modern account.  Specifically, let us
define an ad-hoc version of the:
\begin{definition}[Independence property (e.g., \cite{baldwin98})]
  Let $\calN$ be a set of \emph{unvaried} numerical predicates.  Let
  $\vec{x}, \vec{y}$ be two vectors of first-order variables of size $k$ and
  $\ell$, respectively.  A formula $\phi(\vec{x}, \vec{y})$ of $\FO[\calN]$,
  over a single-letter alphabet, has the \emph{independence property} if for
  all $n > 0$ there are vectors $\vec{a_0}, \vec{a_1}, \ldots, \vec{a_{n-1}}$,
  each of $\bbn^k$, for which for any $M \subseteq [n]$, there is a vector
  $\vec{b_M} \in \bbn^\ell$ such that:\footnote{Note that we evaluate a formula
    over an \emph{infinite} domain; this is well defined in our case since we
    only use \emph{unvaried} predicates and the letter predicates are
    irrelevant.}
  \[\langle \bbn, \vec{x} := \vec{a_i}, \vec{y} := \vec{b_M}\rangle \models
    \phi
    \quad\text{iff}\quad i \in M\enspace.\]
  The logic $\FO[\calN]$ has the \emph{independence property} if it contains
  such a $\phi$.
\end{definition}
Intuitively, a logic has the independence property iff it can encode arbitrary
sets.  Barrington~et~al.\ \cite{baimlascth05}, relying on a deep result of Baldwin and
Benedikt~\cite{baldwin98}, show that:
\begin{theorem}[{\cite[Corollary 4.13]{baimlascth05}}]
  If a logic does not have the independence property, then it has the strong
  Crane Beach Property.
\end{theorem}

We note that this powerful tool cannot show that the logic we consider exhibits
the Crane Beach Property:
\begin{proposition}
  $\FO[\leq, \fin]$ has the independence property.
\end{proposition}
\begin{proof}
  Let $n > 0$, and define $a_i = 2^n+2^i$ for $0 \leq i < n$.  Now for
  $M \subseteq [n]$, let $b_M = 2^n + \sum_{i \in M} 2^i$.  It holds that
  $i \in M$ iff the binary AND of $a_i$ and $b_M$ is $a_i$.  Consider this
  latter binary predicate; its behavior on two arguments that do not share the
  same MSB is irrelevant, and we can thus decide that such inputs are rejected.
  Thanks to this, we obtain a finite-degree predicate.  Consequently, the
  formula that consists of this single predicate has the independence property.
\end{proof}

\section{\(\FO[\leq, \fin]\) has the Crane Beach Property}\label{sec:cbcfin}

\subsection{Communication complexity}

We will show the Crane Beach Property of $\FO[\leq, \fin]$ by a communication
complexity argument.  This approach is mostly unrelated to the use of
communication complexity of~\cite{laut06,chattopadhyay07}; in particular, we are
concerned with two-party protocols with a split of the input in two contiguous
parts, as opposed to worst-case partitioning of the input among multiple
players.  We rely on a characterization of~\cite{fijalkow14} of the class of languages
expressible in \emph{monadic second-order} with varied monadic numerical
predicates.  Writing this class $\MSO[\leq, \mon]$, they state in particular the
following:
\begin{proposition}[{\cite[Theorem~2.2]{fijalkow14}}]\label{prop:msofin}
  Let $L \subseteq A^*$ and define, for all $p \in \bbn$, the equivalence
  relation $\sim_p$ over $A^*$ as: $u \sim_p v$ iff for all $w \in A^p$,
  $u\cdot w \in L \Leftrightarrow v\cdot w \in L$.  If there is a $N \in \bbn$
  such that for all $p \in \bbn$, $\sim_p$ has at most $N$ equivalence classes,
  then $L \in \MSO[\leq, \mon]$.
\end{proposition}

\defmathtext\alice{Alice} \defmathtext\bob{Bob}\defmathtext\neutral{Neutral}
\begin{lemma}\label{lem:comm}
  Let $L \subseteq A^*$.  Suppose there are functions
  $f_\alice \colon A^* \times \bbn \times \binal^* \to \binal$ and
  $f_\bob \colon A^* \times \bbn \times \binal^*$ and a constant $K \in \bbn$
  such that for any $u, v \in A^*$, the sequence, for $1 \leq i \leq K$:

  \vspace{\topsep}
  \noindent\hspace{\IEEElabelindent}%
  \(\begin{array}{@{}l@{\,}lll}
      \text{\labelitemi}\;\; a_i = f_\alice &(u, &|u\cdot v|, &b_1b_2\cdots b_{i-1})\\
      \text{\labelitemi}\;\; b_i = f_\bob &(v, &|u \cdot v|, & a_1a_2\cdots a_i);
    \end{array}\)

  \vspace{\topsep}%
  \noindent is such that $b_K = 1$ iff $u\cdot v \in L$.  Then $L\in \MSO[\leq, \mon]$.
\end{lemma}
\begin{proof}
  We adapt the (folklore) proof that $L$ is regular iff such functions exist
  where $f_\alice$ and $f_\bob$ do not use their second parameter.

  Let $p \in \bbn$.  For any $u \in A^*$, let $c(u)$ be the set of pairs
  $(a_1a_2\cdots a_K, b_1b_2\cdots b_{K-1})$ such that for all
  $1 \leq i \leq K$, it holds that
  $a_i = f_\alice(u, |u| + p, b_1b_2\cdots b_{i-1})$.  Define the equivalence
  relation $\equiv$ by letting $u \equiv v$ iff $c(u) = c(v)$; it clearly has a
  finite number $N = N(K)$ of equivalence classes.  Moreover, if $u \equiv v$
  and $w \in A^p$, then $(u, w)$ and $(v, w)$ define the same sequences of
  $a_i$'s and $b_i$'s, by a simple induction.  Hence
  $u \cdot w \in L \Leftrightarrow v \cdot w \in L$.  This shows that $\equiv$
  refines $\sim_p$, implying, by Proposition~\ref{prop:msofin}, that
  $L \in \MSO[\leq, \mon]$.
\end{proof}

We shall adopt the classical communication complexity view here, and consider
$f_\alice$ and $f_\bob$ as two players, Alice and Bob, that alternate exchanging
a bounded number of bits in order to decide if the concatenation of their
respective inputs is in $L$.  To show that $L$ is in $\MSO[\leq, \mon]$, the
protocol between Alice and Bob should end in a constant number of rounds.  We
will then rely on the fact that:

\begin{figure*}\label{fig:phi}
  \centering
  \begin{tikzpicture}[>=stealth]
  \node (a) {
    \begin{tikzpicture}
      \node (ex) {$\exists x$};
      \node (all) [below= 0.5 of ex] {$\forall y$};
      \node (psi) [below= 0.5 of all] {$\psi$};
      \path[draw] (ex) -- (all);
      \path[draw] (all) -- (psi);
    \end{tikzpicture}
    };
  \node (b) at ($(a) + (7,0.3)$) {
    \begin{tikzpicture}
      \tikzset{bob_style/.style={circle, draw, very thin, color=white!50!black, text=black, inner sep=2pt}}
      \node (or1) {$\lor$};
      \node (ex1) at ($(or1)+(-2,-0.5)$) {$\existsa x$};
      \node (and1) at ($(ex1)+(0,-0.75)$) {$\land$};
      \node (All11) at ($(and1)+(-1,-0.75)$) {$\foralla y$};
      \node (Psi11) at ($(All11)+(0,-0.75)$) {$\psi$};  
      
      \node (All12) at ($(and1)+( 1,-0.75)$) {$\forallb y$};
      \node (Psi12) at ($(All12)+(0,-0.75)$) {$\psi$};  
      
      \path[draw] (or1) -- (ex1);
      \path[draw] (ex1) -- (and1);
      \path[draw] (and1) -- (All11);
      \path[draw] (and1) -- (All12);
      \path[draw] (All11) -- (Psi11);
      \path[draw] (All12) -- (Psi12);

      \node (ex2) at ($(or1)+(2,-0.5)$) {$\existsb x$};
      \path[draw] (or1) -- (ex2);

      \node (and2) at ($(ex2)+(0,-0.75)$)  {$\land$};

      \node (All211) at ($(and2)+(-1,-0.75)$) {$\foralla y$};
      \node (qf211) at ($(All211)+(0,-0.75)$)[align=center] {$\psi$};

      \node (All212) at ($(and2)+(1,-0.75)$) {$\forallb y$};
      \node (qf212) at ($(All212)+(0,-0.75)$)[align=center] {$\psi$};

      \path[draw] (ex2) -- (and2);      
      \path[draw] (and2) -- (All211);
      \path[draw] (and2) -- (All212);
      \path[draw] (All211) -- (qf211);
      \path[draw] (All212) -- (qf212);
      
    \end{tikzpicture}    
  };
  \draw [very thin, color=white!50!black] ($(a)+(-2,-2)$) -- ++(14,0);
  \draw [very thin, color=white!50!black] ($(a)+(2,-2)$) -- ++(0,4);
  
  \node (c) at ($(a)+(5,-6)$) {
    \begin{tikzpicture}
      \tikzset{bob_style/.style={circle, draw, very thin, color=white!50!black, text=black, inner sep=2pt}}
      \node (or1) {$\lor$};
      \node (ex1) at ($(or1)+(6,-1)$) {$\existsb x$};
      \node (and1) at ($(ex1)+(0,-0.75)$) {$\land$};
      \node (All11) at ($(and1)+(1.5,-1.5)$) {$\forallb y$};
      \node (Psi1) at ($(All11)+(0,-2.5)$) {$\psi$};  
      
      \node (All12) at ($(and1)+(- 1.5,-1.5)$) [bob_style] {$\bigwedge$};
      \node (qf11) at ($(All12)+(-0.75,-2.5)$)[align=center] {$\lt{a}(x)\land{}$ \\ $\bot \to \bot$};
      \node (qf12) at ($(All12)+(0.75,-1.5)$) [align=center] {$\lt{a}(x) \land{} $\\ $\bot \to \bot$};
      
      \path[draw] (or1) -- (ex1);
      \path[draw] (ex1) -- (and1);
      \path[draw] (and1) -- (All11);
      \path[draw] (and1) -- (All12);
      \path[draw] (All11) -- (Psi1);
      \path[draw] (All12) -- (qf11) node [scale=0.75, pos=0.5, fill=white]{$y=0$};
      \path[draw] (All12) -- (qf12) node [scale=0.75, pos=0.5, fill=white]{$y=1$};

      \node (ex2) at ($(or1)+(-2,-0.75)$) [bob_style]{$\bigvee$};
      \path[draw] (or1) -- (ex2);

      \node (and21) at ($(ex2)+(-3,-0.75)$)  {$\land$};

      \node (All211) at ($(and21)+(1,-1.5)$) {$\forallb y$};
      \node (qf211) at ($(All211)+(0,-2.5)$)[align=center] {$\top \land{}$ \\ $\top \to \lt{b}(y)$};

      \node (All212) at ($(and21)+(-1.5,-1.5)$)[bob_style] {$\bigwedge$};
      \node (qf2121) at ($(All212)+(-0.75,-2.5)$)[align=center] {$\top  \land{}$ \\ $\bot \to \bot$};
      \node (qf2122) at ($(All212)+( 0.75,-1.5)$)[align=center] {$\top \land{}$ \\ $\top \to \bot$};

      \node (and22) at ($(ex2)+( 2.5,-0.75)$) {$\land$};

      \node (All221) at ($(and22)+(1 ,-1.5)$) {$\forallb y$};
      \node (qf221) at ($(All221)+(0,-1.5)$)[align=center] {$\top \land{}$ \\ $\top \to \lt{b}(y)$};

      \node (All222) at ($(and22)+(-1.5,-1.5)$) [bob_style] {$\bigwedge$};
      \node (qf2221) at ($(All222)+(-0.75,-1.5)$)[align=center] {$\top \land{}$ \\ $\bot \to \bot$};
      \node (qf2222) at ($(All222)+(0.75,-2.5)$)[align=center] {$\top \land{}$ \\ $\bot \to \bot$};
      \path[draw] (ex2) -- (and21) node [scale=0.75, pos=0.5, yshift=0pt, fill=white]{$x=0$};
      \path[draw] (ex2) -- (and22) node [scale=0.75, pos=0.5, yshift=0pt, fill=white] {$x=1$};
      \path[draw] (and21) -- (All211);
      \path[draw] (and21) -- (All212);
      \path[draw] (and22) -- (All221);
      \path[draw] (and22) -- (All222);
      \path[draw] (All211) -- (qf211);
      \path[draw] (All212) -- (qf2121) node [scale=0.75, pos=0.5, fill=white]{$y=0$};
      \path[draw] (All212) -- (qf2122) node [scale=0.75, pos=0.5, fill=white]{$y=1$};
      \path[draw] (All221) -- (qf221);
      \path[draw] (All222) -- (qf2221) node [scale=0.75, pos=0.5, fill=white]{$y=0$};
      \path[draw] (All222) -- (qf2222) node [scale=0.75, pos=0.5, fill=white]{$y=1$};

      \node (Bob) at (4,0.5) [align=center, color=white!40!black] {Quantified by Alice};
      \path[draw,,->,very thin ,color=white!60!black] (Bob.south) .. controls ++(0,-1) and ++(1,0) .. (ex2.east);
      \path[draw,->,very thin ,color=white!60!black] (Bob.south) .. controls ++(0,-3) and ++(1,1) .. (All212.east);
      \path[draw,->,very thin ,color=white!60!black] (Bob.south) .. controls ++(0,-2) and ++(0,3) .. (All222.north);  
      \path[draw,->,very thin ,color=white!60!black] (Bob.south) .. controls ++(0,-0.75) and ++(0,1.5) .. ($(All12.north)+(0,0.75)$)--(All12);

    \end{tikzpicture}
  };
  \node (la) at ($(a)+(-1,1.5)$) {a)};
  \node (lb) at ($(la)+(4,0)$) {b)};
  \node (lc) at ($(la)+(0,-5)$) {c)};

 \end{tikzpicture}
  \caption{The formula $\phi$ as it gets evaluated by Alice and Bob.}
  \label{fig}
\end{figure*}
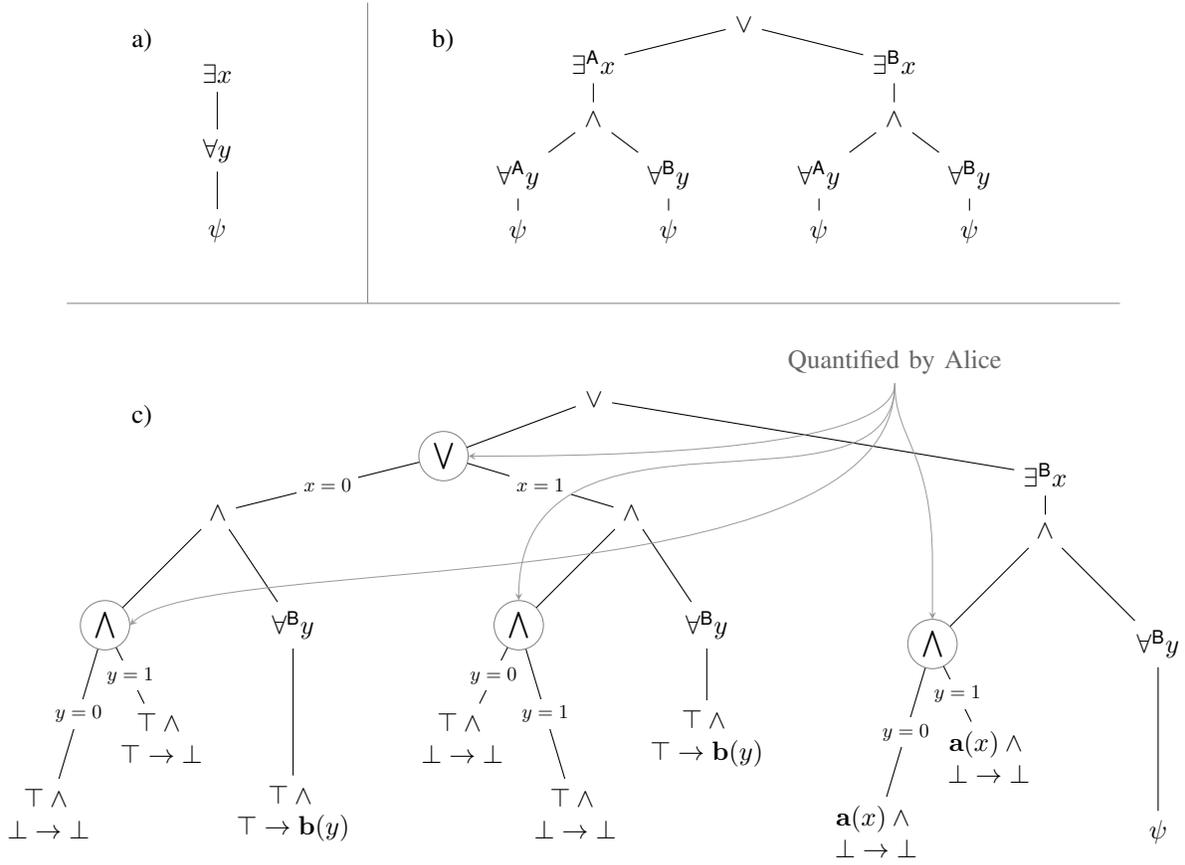

\begin{theorem}[{\cite[Theorem 4.6]{fijalkow14}}]\label{thm:cbcmso}
  $\MSO[\leq, \mon]$ has the Crane Beach Property.
\end{theorem}

\subsection{A toy example: \(\FO[<] \subseteq \MSO[\leq, \mon]\)}

We will demonstrate how the communication complexity approach will be used with
a toy example.  Doing so, the requirements for this protocol to work will be
emphasized, and they will be enforced when showing the Crane Beach Property of
$\FO[\leq, \fin]$ in Section~\ref{sec:cbcfinhere}.

Let us consider the following formula over $A = \{a, b, c\}$:
\[\phi \equiv (\exists x)(\forall y)[\psi],\quad\text{with } \psi \equiv
  \lt{a}(x) \land (x < y \to \lt b(y))\enspace,\]%
depicted as a tree in Figure~\ref{fig}.a.  The formula $\phi$ asserts that
the all the letters after the last $a$ are $b$'s.  In this example, Alice will
receive $u = aa$, and Bob $v = bb$.  Naturally, $\phi$ over words of length $4$
is equivalent to the formula where $\exists x$ is replaced by $\bigvee_{x=0}^3$,
and $\forall y$ is replaced by $\bigwedge_{y=0}^3$; our approach will be to
split this rewriting between Alice and Bob.

Consider the variable $x$.  To check the validity of the formula over a
$u \cdot v$, the variable should range over the positions of both players.  In
other words, the formula is true if there is a position $x$ of Alice verifying
$(\forall y)[\psi]$ \emph{or} a position $x$ of Bob verifying it---likewise for
the universal quantifier.  We thus ``split'' the quantifiers by enforcing the
domain to be either Alice's ($\foralla, \existsa$) or Bob's
($\forallb, \existsb$), obtaining Figure~\ref{fig}.b.

Alice will now expand her quantifiers to range over her word; she will thus
replace, e.g., $(\foralla y)[\psi]$ by $\bigwedge_{y=0}^1 \psi$.  Crucially, at
the leaves of the formula, it is known which variables were quantified by each
player, and if they are Alice's, their values.  Consider for instance a leaf
where Alice substituted $y$ with a numerical value.  The letter predicate
$\lt{b}(y)$ can thus be replaced by its truth value.  More importantly, the
predicate $x < y$ can \emph{also} be evaluated: Either Alice quantified $x$, and
it has a numerical value, or she did not, and we know for sure that $x < y$ does
not hold, since $x$ will be quantified by Bob.  Applied to our example, we
obtain the tree of Figure~\ref{fig}.c.

The resulting formulas at the leaves are thus \emph{free from the variables
  quantified by Alice.}  Moreover, for each internal node of the tree, its
children represent subformulas of bounded quantifier depth, and there are thus a
finite number of possible nonequivalent subformulas.  Once only one subformula
per equivalence class is kept, the resulting tree is of bounded depth and each
node has a bounded number of children.  Hence the size of this tree is bounded
by a value that \emph{only depends on $\phi$}.  Alice can thus communicate this
tree to Bob.  In our example, simplifying the tree, we obtain the formula:
\[(\forallb y)[\lt b(y)] \lor (\existsb x)\Big[\lt a(x) \land (\forallb
  y)[\psi]\Big]\enspace.\]

Finally, Bob can actually quantify his variables, resulting in a formula
with no quantified variable, that he can evaluate, concluding the
protocol.

\emph{Takeaway.} This protocol relies on the fact that predicates that involve
variables from both Alice and Bob can be evaluated by Alice alone.  This enables
Alice to remove ``her'' variables before sending the partially evaluated tree to
Bob, who can quantify the remainder of the variables.

\subsection{The case of \(\FO[\leq, \fin]\)}\label{sec:cbcfinhere}

\begin{theorem}\label{thm:cbcfin}
  $\FO[\leq, \fin]$ has the strong Crane Beach Property.
\end{theorem}
\begin{proof}
  Let $\phi$ be a formula over an alphabet $A$ in $\FO[\leq, \calN]$, for some
  finite subset $\calN$ of $\fin$, and suppose $\phi$ expresses a language $L$
  that admits a neutral letter $e$.  We show that $L \in \MSO[\leq, \mon]$ using
  Lemma~\ref{lem:comm}.  This concludes the proof since by
  Theorem~\ref{thm:cbcmso}, $L$ is a neutral letter regular language in
  $\FO[\arb]$, and it thus belongs to $\FO[\leq]$ (see \cite{baimlascth05}; this
  is essentially a consequence of $\parity \notin \ACz$).

  Let us write $u \in A^*$ for Alice's word, and $v$ for Bob's.  Both players
  will compute a value $N > 0$ that depends solely on $\phi$ and $|u \cdot v|$,
  and the protocol will then decide whether $u\cdot e^N \cdot v \in L$, which is
  equivalent to $u \cdot v \in L$ by hypothesis.  We suppose that a large enough
  $N$ has been picked for the protocol to work, and delay to the end of the
  proof its computation.

  We will henceforth suppose that $\phi$ is given in prenex normal form and that
  all variables are quantified only once:
  \[\phi \equiv (Q_1 x_1)(Q_2 x_2)\cdots(Q_k x_k)[\psi]\enspace,\]
  with $\psi$ quantifier-free and $Q_i \in \{\forall, \exists\}$.  We again see
  formulas as trees with leaves containing quantifier-free formulas.

  Rather than splitting the domain $[|u\cdot e^N \cdot v|]$ at a precise
  position, and tasking Alice to quantify over the first half and Bob over the
  second half, we will rely on a third group, that is ``far enough'' from both
  Alice's and Bob's words.  The core of this proof is to formalize this notion.
  Let us first introduce the tools that will enable this formalization: one set
  of definitions, and two facts that will be used later on.

  \begin{definition}
    Let $C$ be the set of pairs of integers $(p_1, p_2)$ that appear in a same
    tuple of a relation in $\calN$.  Define the \emph{link graph}
    $G = (\bbn, E)$ as the \emph{undirected} graph defined by $(p_1, p_2) \in E$
    iff $p_1 = p_2$ or there are integers $p_1' \leq \{p_1, p_2\} \leq p_2'$
    such that $(p_1', p_2') \in C$.  For $p \in \bbn$, $\lcut(p)$ (\resp
    $\rcut(p)$) is the greatest $q < p$ (\resp smallest $q > p$) which is
    not a neighbor of $p$ in $G$.  Equivalently, $\lcut(p)$ is the smallest
    neighbor of $p$ minus~1, and $\rcut(p)$ is the greatest neighbor of $p$
    plus~1.
  \end{definition}

  Note that $\lcut$ and $\rcut$ are well defined since each vertex of $G$
  has a finite number of neighbors.  This directly implies that:
  \begin{fact}\label{fact:minl}
    The functions $\lcut$ and $\rcut$ are nondecreasing and unbounded.
    Moreover, for any $p \in \bbn$, $\lcut(p) < p < \rcut(p)$.
  \end{fact}



  Writing $\rcut^n$ for the function $\rcut$ composed $n$ times with
  itself, and similarly for $\lcut$, we have:
  \begin{fact}\label{fact:move}
    For any position $p$ and $n > m \geq 0$:
    \begin{itemize}
    \item $\lcut^m(\rcut^n(p)) \geq \rcut(p)$;
    \item $\rcut^m(\lcut^n(p)) \leq \lcut(p)$.
    \end{itemize}
  \end{fact}
  \begin{proof}
    This is easily shown by induction; we prove the first item, the second being
    similar.  For $n = 1$, this is clear.  Let $n > 1$.  If $m = 0$, this is
    immediate from Fact~\ref{fact:minl}, let thus $m > 0$.  We have that:
    \[\lcut^m(\rcut^n(p)) =
      \lcut(\lcut^{m-1}(\rcut^{n-1}(p')))\enspace,\]
    with $p' = \rcut(p)$.  By induction hypothesis and the fact that
    $\lcut$ is nondecreasing, it holds that:
    \[\lcut^m(\rcut^n(p)) \geq \lcut(\rcut(p')) = q\enspace.\]
    Let $p'' = \rcut(p')$.  By definition of $\lcut$, $(q+1, p'')$ is an edge in
    $G$.  Now by definition of $G$, if $q < p'$, then $(p', p'')$ should also be
    an edge in $G$, which contradicts the definition of $p''$.  Hence
    $q \geq p'$, showing the property.%
    \addtoqedsymbol{Fact~\ref{fact:move}}
  \end{proof}

  Let us now suppose we have two large positions
  $|u| \ll \ell_0 \ll r_0 \ll |v|$, the requirements on which will be made clear
  shortly.  Let us deem a position $p$ to be \emph{Alicic} if $p \leq \ell_0$,
  \emph{Bobic} if $p \geq r_0$, and \emph{Neutral} otherwise; we call this the
  \emph{type} of the position.  We wish to ensure that two positions of two
  different types cannot be linked in $G$, so that they cannot appear in a tuple
  of a predicate in $\calN$.  This surely is not the case if the typing of
  positions does not reflect previously typed positions, e.g., $\ell_0 - 1$ is
  Alicic, but $\ell_0$ is Neutral, and their distance may not be large enough to
  ensure that they do not form an edge in $G$.  Thus the boundaries of the
  zones, $\ell_0$ and $r_0$, will be moving with each new typing.  Formally, let
  $T = \{\alice, \neutral, \bob\}$ be an alphabet, and define the function
  $\bounds\colon T^{\leq k} \to [|u\cdot e^N\cdot v|]^2$ by:
  \begin{align*}
    \bounds(\eps) & = (l_0, r_0)\\
    \bounds(t_1t_2\cdots t_i) & =\mbox{\hspace{5cm}}\\
    \mathclap{\mbox{\hspace{4cm}}\begin{cases}
                                  (\rcut^n(\ell), r) & \text{if } t_i =
                                  \alice\\
                                  (\lcut^n(\ell), \rcut^n(r)) & \text{if } t_i
                                  = \neutral\\
                                  (\ell, \lcut^n(r)) & \text{if } t_i = \bob
                                \end{cases}}\\
    \mathclap{\mbox{\hspace{3cm}}\text{with } (\ell, r) = \bounds(t_1t_2\cdots
    t_{i-1}) \text{ and } n = 2^{k-i}.}
  \end{align*}

  \begin{assumption}
    We henceforth assume that if $(\ell, r) = \bounds(h)$ for some word
    $h \in T^{\leq k}$, then $|u|< \ell < r < |u|+N$.  This will have to be
    guaranteed by carefully picking $N$, $\ell_0$ and $r_0$.
  \end{assumption}

  The type of a position $p$ \emph{under type history}
  $t_1t_2\cdots t_i \in T^*$ is computed by first taking
  $(\ell, r) = \bounds(t_1t_2\cdots t_i)$, and reasoning as before: it is Alicic
  if $p \leq \ell$, Bobic if $p \geq r$, and Neutral otherwise.  This is well
  defined since $\ell < r$ by our Assumption.  The crucial property here is as
  follows:
  \begin{fact}\label{fact:crux}
    Let $p_1, p_2, \ldots, p_k$ be positions, and inductively define the type
    $t_i$ of $p_i$ as its type under type history $t_1t_2\cdots t_{i-1}$.
    \begin{enumerate}
    \item Two positions with different types do not form an edge in~$G$;
    \item All Alicic positions are strictly smaller than the Neutral ones, which
      are strictly smaller than the Bobic ones;
    \item All Neutral positions are labeled with the neutral letter.
    \end{enumerate}
  \end{fact}
  \begin{proof}
    \emph{(Points 1 and 2.)}\quad Suppose $p_i$ is Alicic and $p_j$ is Neutral,
    with $i < j$.  Let $(\ell, r) = \bounds(t_1t_2\cdots t_{i-1})$, we thus have
    that $p_i$ is maximally $\ell$.  Let
    $(\ell', r') = \bounds(t_1t_2\cdots t_{j-1})$, then $p_j$ is minimally
    $\ell'+1$.  By definition, once the types of $p_1, p_2, \ldots, p_i$ are
    fixed, the smallest $\ell'$ that can be obtained with the types $t_{> i}$ is
    by having all positions $p_t$, with $i < t < j$, Neutral.  In that case, an
    easy computation shows that $\ell'$ would be:
    \[\lcut^{2^{k-(j-1)}}(\lcut^{2^{k-(j-2)}}(\cdots
      (\lcut^{2^{k-(i+1)}}(\rcut^{2^{k-i}}(\ell)))\cdots))\enspace.\]
    That is, $\lcut$ is composed with itself $m$ times with:
    \begin{align*}
      m = 2^{k-(i+1)}+ \cdots + 2^{k-(j-1)} & < \sum_{s = i+1}^k 2^{k-s}\\
      & < 2^{k-i} = n\enspace.
    \end{align*}
    Hence $\ell'$ is at most $\lcut^m(\rcut^n(\ell))$ with $m < n$, and
    by Fact~\ref{fact:move}, $\ell' \geq \rcut(\ell)$.  Hence $(p_i, p_j)$
    is not an edge in $G$, and $p_i < p_j$.

    The other cases are similar.  For instance, if $p_i$ is Neutral and~$p_j$
    Bobic, with $i < j$, then, with the same notation as above, $\ell'$ can be
    at most $\lcut^m(\rcut^n(\ell))$, and by Fact~\ref{fact:move},
    $\ell' \geq \lcut(\ell)$.

    \emph{(Point 3.)}  This is a direct consequence of the Assumption.
    Consider $(\ell, r) = \bounds(\neutral^k)$; this provides
    the minimal $\ell$ and maximal $r$  between which a position can be labeled
    Neutral.  By the Assumption, $|u| < \ell < r < |u| + N$, hence a Neutral
    position has a neutral letter.
    \addtoqedsymbol{Fact~\ref{fact:crux}}
  \end{proof}

  We are now ready to present the protocol.  First, we rewrite
  quantifiers using Alicic/Neutral/Bobic \emph{annotated} quantifiers:
  \begin{itemize}
  \item $(\forall x)[\rho] \leadsto (\foralla x)[\rho] \;\land\; (\foralln
    x)[\rho] \;\land\; (\forallb x)[\rho]$,
  \item
    $(\exists x)[\rho] \leadsto (\existsa x)[\rho] \;\lor\; (\foralln x)[\rho] \;\lor\;
    (\existsb x)[\rho]$.
  \end{itemize}
  Let us further equip each node with the type history of the variables
  quantified before it; that is, each node holds a string
  $t_1t_2\cdots t_n \in T^{\leq k}$ where $t_i$ is the annotation of the $i$-th
  quantifier from the root to the node, excluding the node itself.

  Now if we were given the entire word $u \cdot e^N \cdot v$, a way to evaluate
  the formula that respects the semantic of ``Alicic'', ``Neutral'', and
  ``Bobic'' is as follows:

  \algblockdefx[foreach]{Foreach}{EndForeach}[1]{\textbf{foreach} #1 \textbf{do}}{\textbf{end}}
  \begin{algorithm}
    \caption{Formula Evaluation}
    \begin{algorithmic}[1]
      \Foreach{quantifier node $\foralla x$ or $\existsa x$}
        \State $(\ell, r) := \bounds(\text{type history at node})$
        \If{node is $\foralla x$}
          \State Replace node with $\bigwedge_{x=0}^{\ell}$
        \EndIf
        \State Similarly with $\exists$ becoming $\bigvee$
      \EndForeach  
      \State Evaluate the part of the leaves than can be evaluated
      \Foreach{quantifier node}
        \State $(\ell, r) := \bounds(\text{type history at node})$
        \If{node is $\foralln x$}
          \State Replace node with $\bigwedge_{x=\ell+1}^{r-1}$
        \ElsIf{node is $\forallb x$}
          \State Replace node with $\bigwedge_{x=r}^{|ue^Nv|}$
        \EndIf
        \State Similarly with $\exists$ becoming $\bigvee$
      \EndForeach
      \State Finish evaluating the tree
    \end{algorithmic}
    \label{alg}
  \end{algorithm}

  This is precisely the algorithm that Alice and Bob will execute.  First, Alice
  will quantify her variables according to the $\bounds$ of the type history of
  each node, as in Algorithm~\ref{alg}.  At the leaves, she will thus obtain the
  formula $\psi$, and have a set of quantified Alicic variables.  She can then
  evaluate $\psi$ \emph{partially}: if an atomic formula only relies on Alicic
  variables, she can compute its value.  If an atomic formula uses a mix of
  Alicic and non-Alicic variables, then she can \emph{also} evaluate it: if the
  formula is a numerical predicate, then by Fact~\ref{fact:crux}.1, it will be
  valued \emph{false}; if the formula is of the form $x < y$, then it is true
  iff $x$ is Alicic, by Fact~\ref{fact:crux}.2.  Alice now simplifies her tree:
  logically equivalent leaves with the same parent are merged, and inductively,
  each internal node keeps only a single occurrence per formula appearing as a
  child.  We remark that the semantic of the tree is preserved.  This results in
  a tree whose size depends \emph{solely} on $\phi$, and the values of $N$,
  $\ell_0$, and $r_0$, and Alice can thus send it to Bob.

  Bob will now expand the remaining quantifiers (Neutral and Bobic), respecting
  the bounds of the type history, as in Algorithm~\ref{alg}.  He can then
  evaluate all the leaves, since, by Fact~\ref{fact:crux}.3, the only letter
  predicate true of a Neutral position is that of the neutral letter.  This
  concludes the protocol, which clearly produces the same result as
  Algorithm~\ref{alg}.

  \def\tmin{\text{min}} \def\tmax{\text{max}}%
  \emph{What are $N$, $\ell_0$, $r_0$?} We check that Alice and Bob can agree on
  these values without communication.  The requirements were made explicit in
  our Assumption.  The values computed by the function $\bounds$ are obtained by
  applying $\lcut$ and $\rcut$ on $\ell_0$ and $r_0$ \emph{at most}
  $n = \sum_{i=0}^{k-1} 2^i$ times.  From Fact~\ref{fact:minl}, it is clear that
  any $(\ell, r) = \bounds(h)$, for $h \in T^{\leq k}$, verifies:
  \begin{itemize}
  \item $\ell_\tmin = \lcut^n(\ell_0) \leq \ell \leq \rcut^n(\ell_0) = \ell_\tmax$;
  \item $r_\tmin = \lcut^n(r_0) \leq r \leq \rcut^n(r_0) = r_\tmax$.
  \end{itemize}
  Hence we pick $\ell_0 = \rcut^{n+1}(|u|)$, ensuring, by Fact~\ref{fact:move},
  that $\ell_\tmin > |u|$.  Next, we pick $r_0$ to be $\rcut^{n+1}(\ell_\tmax)$,
  ensuring that $r_\tmin > \ell_\tmax$ by the same Fact~\ref{fact:move}.
  Finally, we pick $N = \rcut^{n+1}(r_0)$, ensuring, by Fact~\ref{fact:minl},
  that $N > r_\tmax$, so that in particular, $r_\tmax < |u| + N$.  We then
  indeed obtain that $|u| < \ell < r < |u| + N$, as required.  Note that these
  computations depend \emph{solely} on $\phi$ and the lengths of $u$ and~$v$.%
  \addtoqedsymbol{Theorem~\ref{thm:cbcfin}}
\end{proof}

\begin{remark}
  It should be noted that the crux of this proof is that a relation $R(x, y)$
  with $x$ Alicic and $y$ Neutral or Bobic can be readily evaluated by Alice.
  If $R$ were monadic, then it could not mix two positions of different types,
  hence Alice could still remove all of her variables at the end of her
  evaluation.  The rest of the protocol will be similar, with Bob quantifying
  the remaining positions.  This shows that $\FO[\leq, \mon, \fin]$ also has the
  Crane Beach Property.
\end{remark}

\section{On counting}\label{sec:count}

A compelling notion of computational power, for a logic, is the extent to which
it is able to precisely evaluate the \emph{number of positions} that verify a
formula.  This is formalized with the following standard definition:
\begin{definition}
  For a nondecreasing function $f(n) \leq n$, a logic is said to \emph{count up
    to $f(n)$} if there is a formula $\phi(c)$ in this logic such that for all
  $n$ and $w \in \binal^n$:
  \[w \models \phi(c) \quad \Leftrightarrow \quad c \leq f(n) \land c = \text{number of 1's
      in } w\enspace.\]
\end{definition}

It is known from~\cite{baimlascth05} that if a logic can count up to
$\log(\log(\cdots(\log n)))$, for some number of iterations of $\log$, then the
logic does not have the Crane Beach Property.  It has also been
conjectured~\cite{baimlascth05,lee03} that a logic has the Crane Beach Property iff it
cannot count beyond a constant.  It is not known whether there exists a set of
predicates $\calN$ such that $\FO[\calN]$ can count beyond a constant but not up
to $\log n$.

We define a much weaker ability:
\begin{definition}
  For a nondecreasing function $f(n) \leq n$, a logic is said to \emph{sum
    through $f(n)$} if there is a formula $\phi(a, b, c)$ in this logic such
  that for all $n$ and $w \in \binal^n$:
  \[w \models \phi(a, b, c) \quad \Leftrightarrow \quad a = b + f(c)\enspace.\]
\end{definition}

This is in general even weaker than being able to sum ``up to'' $f(n)$, that is,
having a formula expressing that $a = b + c$ and $c \leq f(n)$.  Naturally,
counting and summing are related:
\begin{proposition}\label{prop:counttosum}
  Let $\calN$ be a set of unvaried numerical predicates.  If $\FO[\leq, \calN]$
  can count up to $f(n)$, it can sum through $f(n)$.
\end{proposition}
\begin{proof}
  Letting $\phi(c)$ be the formula that counts up to $f(n)$, we modify it into
  $\phi'(a, b, c)$ by changing the letter predicates to consider that there is a
  $1$ in position $p$ iff $b \leq p < a$.  This expresses that $a = b + c$
  provided that $c \leq f(n)$.

  Next, the graph $F$ of $f$ is obtained as follows.  First, modify $\phi(c)$
  into $\phi'(c, c')$, by restricting all quantifications to $c$ and replacing
  the letter predicates to have $1$'s in all positions below $c'$.  Second,
  $(c, c') \in F$ iff $c'$ is maximal among those that verify $\phi'(c, c')$.
  This relies on the fact that $\calN$ consists solely of unvaried predicates.

  The logic can then sum through $f(n)$ by:
  \[\psi(a, b, c) \equiv (\exists c')[F(c, c') \land a = b + c']\enspace.\qedhere\]
\end{proof}

\begin{remark}
  Proposition~\ref{prop:counttosum} depends crucially on the fact that the
  predicates are unvaried to show that the graph of the summing function is
  expressible.  Writing $\calS$ for the set of varied monadic predicates
  $S = (S_n)_{n\geq 0}$ with $|S_n| = 1$ for all $n$, it is easily shown that
  $\FO[\leq, +, \times, \calS]$ can count up to any function $\leq \log n$.
  However, we conjecture that there are functions whose graphs are not
  expressible in this logic.
\end{remark}

\begin{proposition}\label{prop:count}
  $\FO[\leq, \fin]$ cannot sum through beyond a constant.
\end{proposition}
\begin{proof}
  Suppose for a contradiction that $\FO[\leq, \fin]$ can sum through a
  nondecreasing unbounded function $f$ using a formula $\phi(a, b, c)$.  Let
  $\BIT$ be the binary predicate true of $(x, y)$ if the $y$-th bit of $x$ is
  $1$.  We define a translated version as:
  \[\BIT' = \{(x, y) \mid (x, y - f(x)) \in \BIT\}\enspace.\]
  We show that $\BIT'$ is of finite degree.  Let $n \in \bbn$, and suppose $(n, y)
  \in \BIT'$.  This implies in particular that $0 < y - f(n) < \log n$, hence $n$
  appears a finite number of time as $(n, y)$ in $\BIT'$.  Suppose $(x, n) \in
  \BIT'$, then $n - f(x) > 0$, but for $x$ large enough, $f(x) > n$, hence there
  can only be a finite number of pairs $(x, n)$ in $\BIT'$.

  Now $\BIT$ can be defined in $\FO[\leq, \fin]$ using $\phi$, since
  $\BIT(x, y)$ holds iff $(\exists z)[\phi(z, y, x) \land \BIT'(x, z)]$, a
  contradiction concluding the proof.
\end{proof}

\begin{corollary}
  $\FO[\leq, \fin]$ cannot count beyond a constant.
\end{corollary}

\section{Conclusion}
We showed that $\FO[\leq, \fin]$ is one simple predicate away from expressing
all of $\FO[\arb]$, and that it exhibits the Crane Beach Property.  This logic
is thus really on the brink of a crevice on the Crane Beach, and exemplifies a
diverse set of behaviors that fit the intuition that neutral letters should
render numerical predicates essentially useless.  We emphasize some future
research directions:
\begin{itemize}
\item As a consequence of our results, one can show that a nonregular neutral
  letter language $L$ is not in $\ACz$ as follows.  Assume $L \in \ACz$ for a
  contradiction, and let $\phi \in \FO[\leq, \MSBz, \fin]$ be a formula
  expressing it.  Suppose that one can show that $\phi$ can be rewritten
  \emph{without} the predicate $\MSBz$, then $L \in \FO[\leq, \fin]$, and thus
  $L$ is regular, a contradiction.  We hope to be able to apply this strategy in
  the future.
\item As noted in~\cite{roy-straubing07} and~\cite{baimlascth05} and studied in
  particular in~\cite{krebs12}, the interest in circuit complexity calls for the
  study of logics with more sophisticated quantifiers, notably modular
  quantifiers and, more generally, monoidal quantifiers.  Hence the natural
  question here is whether $\FOMOD[\leq, \fin]$ has the Crane Beach Property.
\item As asked in~\cite{baimlascth05}, can we dispense from our implicit reliance
  on the lower bound $\parity \notin \ACz$?  In the cases of~\cite{baimlascth05},
  and as noted by the authors, this would be very difficult, as their results
  \emph{imply} the lower bound.  Here, the strong Crane Beach Property for
  $\FO[\leq, \fin]$ does not directly imply the lower bound.  To show that
  $\parity \notin \ACz$, one could \emph{additionally} prove that all the
  regular, neutral letter languages of $\FO[\leq, \MSBz, \fin]$ are in
  $\FO[\leq, \fin]$---we know that this statement holds, but only thanks to
  $\parity \notin \ACz$.
\item Are we really on the brink of falling off the Crane Beach?  That is, are
  there unvaried predicates that cannot be expressed in $\FO[\leq, \fin]$ but
  can still be added to the logic while preserving the Crane Beach Property?  We
  noted that all \emph{varied} monadic predicates can be added safely, but
  already very simple predicates falsify the Crane Beach Property.  For
  instance, with $F$ the graph of the 2-adic valuation, $\FO[\leq, F]$ is as
  expressive as $\FO[\leq, +, \times]$ (see \cite[Theorem 3]{schwentick98}),
  which does not have the Crane Beach Property~\cite{baimlascth05}.
\item Numerical predicates correspond in a precise sense~\cite{bela06} to the
  computing power allowed to construct circuit families for a language.  Is
  there a natural way to present $\FO[\leq, \fin]$-uniform circuits?
\end{itemize}

\section*{Acknowledgment}

The authors would like to thank Thomas Colcombet, Arnaud Durand, Andreas Krebs,
and Pierre McKenzie for enlightening discussions, and Michael Blondin for his
careful proofreading.

\nocite{straubing94}

\bibliographystyle{IEEEtran}
\bibliography{bib}

\end{document}